\documentclass[letterpaper, 10 pt, conference]{ieeeconf}
\IEEEoverridecommandlockouts                              
\usepackage{graphics} 
\usepackage{epsfig} 
\usepackage{mathptmx} 
\usepackage{times} 
\usepackage{amsmath} 
\usepackage{amssymb}  
\usepackage{subcaption}
\usepackage[english]{babel}
\newtheorem{theorem}{Theorem}[section]

\newtheorem{assumption}{Assumption}

\title{\LARGE \bf Critic-Only Integral Reinforcement Learning Driven by Variable Gain Gradient Descent for Optimal Tracking Control}

\author{Amardeep Mishra and Satadal Ghosh
\thanks{Amardeep Mishra is a Research Scholar at Department of Aerospace Engineering, IIT Madras, Chennai-600036, India
        {\tt\small ae15d405@smail.iitm.ac.in}}%
\thanks{Satadal Ghosh is Faculty at Department of Aerospace Engineering, IIT Madras, Chennai-600036, India 
        {\tt\small satadal@iitm.ac.in}}%
}
\begin{document}

\maketitle
\thispagestyle{empty}
\pagestyle{empty}


\begin{abstract}  :              
Integral reinforcement learning (IRL) was proposed in literature to obviate the requirement of drift dynamics in adaptive dynamic programming framework. 
Most of the online IRL schemes in literature require two sets of neural network (NNs), known as actor-critic NN and an initial stabilizing controller.
Recently, for RL-based robust tracking this requirement of initial stabilizing controller and dual-approximator structure could be obviated by using a modified gradient descent-based update law containing a stabilizing term with critic-only structure.
To the best of the authors' knowledge, there has been no study on leveraging such stabilizing term in IRL algorithm framework to solve optimal trajectory tracking problems for continuous time nonlinear systems with actuator constraints. To this end a novel update law leveraging the stabilizing term along with variable gain gradient descent in IRL framework is presented in this paper. With these modifications, the IRL tracking controller can be implemented using only critic NN, while no initial stabilizing controller is required. Another salient feature of the presented update law is its variable learning rate, which scales the pace of learning based on instantaneous Hamilton-Jacobi-Bellman error and rate of variation of Lyapunov function along the system trajectories. The augmented system states and NN weight errors are shown to possess uniform ultimate boundedness (UUB) stability under the presented update law and achieve a tighter residual set.
This update law is validated on a full 6-DoF nonlinear model of UAV for attitude control. 
\end{abstract}



\section{Introduction}
In order to find optimal control policies for a generic continuous time dynamic system, solving Hamilton-Jacobi-Bellman (HJB) equation is essential. Neural networks (NNs) have been utilized to provide approximate solution to  (HJB) equation in adaptive dynamic programming (ADP) framework \cite{abu2005nearly}, \cite{vamvoudakis2010online}, \cite{liu2013adaptive}.
ADP algorithm for optimal tracking control problem (OTCP) for continuous time nonlinear system (CTNS) was initially proposed by devising transient and steady-state controllers \cite{zhang2011data}. 
However, the major limitation of this method was that it required the control gain matrix to be invertible in order to implement the steady state controller. 
This limitation was bypassed in \cite{modares2014optimal} by proposing an augmented system comprising of error and desired dynamics. 
 
The schemes discussed above require an initial stabilizing control for policy iteration and dual-approximator - actor-critic NN - to implement reinforcement learning (RL). Finding an initial stabilizing controller is not a trivial task. To address these, a single-approximator (critic-only NN) structure, in which a stabilizing term was included in the critic update law to relax the criterion of initial stabilizing control for ADP algorithm, was proposed for optimal regulation \cite{liu2015reinforcement} and tracking in \cite{dierks2010optimal} and robust tracking in \cite{yang2015robust}.

Most of the schemes discussed above except \cite{liu2013adaptive} and \cite{modares2014optimal} require full knowledge of plant dynamics, which may not be available in many real-life applications. In recent times Integral RL (IRL) schemes have been proposed in literature as an alternate formulation of Bellman equation to obviate the requirement of drift dynamics. Since exploration phase of off-policy IRL formulations make them unsuitable for various engineering applications such as control of aerial vehicles and  the final control policies learnt at the end of this exploration phase  cannot effectively cope with a different operational set-up, online and on-policy IRL algorithm is considered in this paper. 

Synchronous tuning of actor-critic NN based on a gradient descent-driven update law in on-policy IRL framework was considered for optimal regulation \cite{vamvoudakis2014online} and optimal tracking \cite{modares2014optimal}.
However, in both of these papers, it was required to have an initial stabilizing controller to initiate the process of policy iteration, which is non-trivial. Also, the dual-approximation structure they used increased the computational load. On the other hand, \cite{liu2015reinforcement} and \cite{yang2015robust} did not require an initial stabilizing controller for their policy iteration algorithm. However, they did not utilize IRL framework and hence needed full knowledge of the nominal plant dynamics.

Inspired by \cite{modares2014optimal}, \cite{liu2015reinforcement} and \cite{yang2015robust}, this paper addresses those concerns by proposing a novel update law in IRL framework to solve OTCP for CTNS with actuator constraints by utilizing a variable gain gradient descent \cite{mishra2019variable} that requires neither an initial stabilizing controller nor drift dynamics. 
Moreover, unlike \cite{modares2014optimal} and  \cite{vamvoudakis2014online}, 
this paper uses critic-only NN for both policy evaluation and policy iteration thereby reducing the computational load. The primary contributions of this paper are as follow.
\begin{itemize}
    \item Development of a novel update law for single NN (critic-only) in IRL framework to solve OTCP for CTNS with actuator constraints that does not require an initial stabilizing controller.
    \item The presented critic update law is driven by variable gain gradient descent that can adjust the learning rate depending on the instantaneous HJB error and the instantaneous rate of variation of Lyapunov function.
    \item The added benefit of the proposed variable gain gradient descent is that it shrinks the size of the residual set on which the augmented system trajectories converge to.
\end{itemize}
The rest of the paper is organized as follows, Section \ref{prelim} introduces OTCP for CTNS with actuator constraints and approximation of value function using a single NN. Section \ref{vargain} presents the prime contribution of this paper, i.e., the variable gain gradient descent-based online parameter update law in IRL framework for OTCP and Section \ref{res} validates the performance of the presented update law on a full 6-degrees-of-freedeom (6-DoF) nonlinear model of UAV, and concluding remarks are finally presented in Section \ref{conclusion}.

\section{Optimal tracking control problem and value function approximation}\label{prelim}
\subsection{Preliminaries}

Let the control-affine system dynamics be defined by:
\begin{equation}
\dot{x}=f(x)+g(x)u
\label{eq:affine}
\end{equation}
where, $x \in \mathbb{R}^n$ is the state of the system, $u \in \mathbb{R}^m$ is the control vector, 
$f(x):~ \mathbb{R}^n \to \mathbb{R}^n$ is the drift dynamics and $g(x):~ \mathbb{R}^{n} \to \mathbb{R}^{n\times m}$ is the input coupling dynamics. 
\begin{assumption}\label{uncert}
\textnormal{The system dynamics \eqref{eq:affine} is Lipschitz continuous in $x$ on compact set $\Omega \subseteq \mathbb{R}^n$. 
Further, $\exists L_f>0, g_M>0 \ni \|f(x)\| \leq L_f\|x\|,~0<\|g(x)\|<g_M,~ \forall x \in \mathbb{R}^n$. }
\end{assumption}
\begin{assumption}\label{ref_dyn}
\textnormal{Let $x_d(t)$ be the desired bounded reference trajectory governed by $\dot{x}_d(t)=H(x_d(t)) \in \mathbb{R}^n$ and $H(0)=0$, where $H(x_d(t))$ is Lipschitz continuous in $x_d$.}
\end{assumption}
An augmented system is formulated in this section, and a discounted cost function with actuator constraints for this augmented system is considered. 
Tracking error dynamics is defined as, $\dot{e}=\dot{x}-\dot{x}_d=f(x_d+e)+g(x_d+e)u(t)-H(x_d(t))$.
Therefore, the dynamics of augmented system states $z=[e^T,x_d^T]^T$ can be compactly written as:
\begin{equation}
\dot{z}=F(z)+G(z)u
\label{aug}
\end{equation}
where, $u \in \mathbb{R}^m$, $F:\mathbb{R}^{2n} \to \mathbb{R}^{2n}$ and $G: \mathbb{R}^{2n} \to \mathbb{R}^{2n\times m}$ are given by:
\begin{equation}
\begin{split}
F(z)=\begin{pmatrix}
f(e+x_d)-H(x_d) \\
H(x_d)
\end{pmatrix},G(z)=
\begin{pmatrix}
g(e+x_d) \\
0
\end{pmatrix}
\end{split}
\end{equation}
From Assumptions \ref{uncert} and \ref{ref_dyn}, it can be observed that, $\|F(z)\| \leq L_F\|z\|$ and $\|G(z)\| \leq g_M$, where $L_F \geq 0$.

The infinite horizon discounted cost function for (\ref{aug}) is considered as follows,
\begin{equation}
V(z)=\int_t^{\infty} e^{-\gamma(\tau-t)}[Q(z)+U(u)]d\tau    
\label{cost}
\end{equation}
where, $Q(z)=z^TQ_1z$ and $Q_1 \in \mathbb{R}^{2n\times 2n}$ is a positive definite matrix. 
The term $U(u)$ above is for control constraints. Following \cite{abu2005nearly} and  \cite{lyashevskiy1996constrained}, 
in this paper, 
$U(u)$ is defined as,
\begin{equation}
\footnotesize
\begin{split}
U(u)=2u_m\int_0^u (\psi^{-1}(\nu/u_m))^TRd\nu = 2u_m\sum_{i=1}^{m}\int_0^{u_i} (\psi^{-1}(\nu_i/u_m))^TR_i d\nu_i
\end{split}
\end{equation}
where, $R \in \mathbb{R}^{m\times m}$ is a positive definite diagonal matrix, ($\psi: \mathbb{R}^m \to \mathbb{R}^m$) is a function possessing following properties -(i) It is odd and monotonically increasing. (ii) It is a smooth bounded function possessing, ($|\psi(.)| \leq 1$). 
In literature $\psi$ has been considered as $\tanh,erf,sigmoid$ functions. 
In this paper, $\psi=\tanh{(.)}$.
This also ensures $U(u)$ to remain positive. The discount factor, $\gamma \geq 0$, defines the value of utility in future. 
Differentiating (\ref{cost}) along the system trajectories and rearranging the terms,
\begin{equation}
\small
\begin{split}
 \nabla_z{V}(F(z)+G(z)u)-\gamma V(z)+z^TQ_1z+U(u)
 =\mathcal{H}(z,u,\nabla_z{V})=0   
\end{split}
\label{Ha}
\end{equation}
\vspace{-.1cm}
where, $\mathcal{H}(.)$ represents the Hamiltonian.
Let $V^*(z)$ be the optimal cost function satisfying $\mathcal{H}(.)=0$ and is given by,
\begin{equation}
V^*(z)=\min_{u} \int_t^{\infty}e^{-\gamma(\tau-t)}[z^TQ_1z+U(u)]d\tau
\label{opt_cost}
\end{equation}
Thus, $\mathcal{H}(.)=0$ can be re-written in terms of $V^*(z)$ as,
\begin{equation}
 \nabla_z{V^*}(F(z)+G(z)u)-\gamma V^*(z)+z^TQ_1z+U(u)=0
\label{hjb1} 
 \end{equation}
Differentiating (\ref{hjb1}) with respect to $u$, i.e, $\partial \mathcal{H}/ \partial u=0$, closed-form of optimal control action $u^*$ is obtained as,
\begin{equation}
u^*=-u_m\tanh{\Big(R^{-1}G(z)^T\nabla_z{V^*}/2u_m\Big)}
\label{opt_u}
\end{equation}
From (\ref{hjb1}) and (\ref{opt_u}), the HJB equation is rewritten as,
\begin{equation}
\begin{split}
V^*_zF(z)-2u_m^2A^T(z)R\tanh(A(z))+z^TQ_1z+\\
2u_m\int_0^{u^*}\tanh^{-1}(\nu/u_m)^TRd\nu-\gamma V^*=0
\end{split}
\label{hjb_int}
\end{equation}
where, $V_z^*\triangleq\nabla_z{V^*}$, and $A\triangleq(1/2u_m)R^{-1}G(z)^TV^*_z \in \mathbb{R}^m$.
$U(u)$ can be simplified as:
\begin{equation}
\begin{split}
&U(u)=2u_m\int_0^{-u_m\tanh{A(z)}}\tanh^{-1}(\nu/u_m)^TRd\nu\\
&=2u_m^2A^T(z)R\tanh{A(z)}+
u_m^2\sum_{i=1}^{m}R_i\log[1-\tanh^2{A_i(z)}]
\end{split}
\label{hjb_int1}
\end{equation}
\subsection{Approximation of value function}
In this subsection, a single NN structure is utilized to approximate the value function. 
Leveraging Weierstrass approximation theorem \cite{abu2005nearly}, any smooth nonlinear mapping can be approximated by selecting a NN with sufficient number of nodes in the hidden layer.
Following this, an ideal NN weight vector $W\in \mathbb{R}^{N_1}$ that can approximate value function is considered,
\begin{equation}
V^*(z)=W^T\vartheta+\epsilon(z)
\label{appro}
\end{equation}
where, $\vartheta: \mathbb{R}^{2n} \to \mathbb{R}^{N_1}$ ($N_1$ being the number of nodes in hidden layer) is the regressor vector for critic NN. Then, gradient of $V^*(z)$ can be expressed as:
\begin{equation}
V^*_{z}=\nabla\vartheta^T W+\nabla\epsilon(z)
\label{vstar}
\end{equation}
\vspace{-.1cm}
Using (\ref{vstar}) in (\ref{opt_u}),
\begin{equation}
\begin{split}
u^*=-u_m\tanh{\Big(\frac{1}{2u_m}R^{-1}G(z)^T \nabla{\vartheta}^TW+\epsilon_{uu}\Big)}
\end{split}
\label{ustar}
\end{equation}
where, $\small \epsilon_{uu}=(1/2u_m)R^{-1}G^T(z)\nabla{\varepsilon}(z)=\small [\varepsilon_{{uu}_{11}},\varepsilon_{{uu}_{12}},...,\varepsilon_{{uu}_{1m}}]^T \in \mathbb{R}^m$.
Using mean value theorem, it can be re-written as $u^*=-u_m\tanh{(\tau_1(z))}+\epsilon_{u^*}    $

where, $\tau_{1}(z)=(1/2u_m)R^{-1}\hat{G}(z)^T \nabla{\vartheta}^TW=[\tau_{11},...,\tau_{1m}]^T \in \mathbb{R}^m$ and 
$\epsilon_{u^*}=-(1/2)((I_m-diag(\tanh^2{(\xi)}))\hat{G}^T(z)\nabla\epsilon)$ with $\xi \in \mathbb{R}^m$ and $\xi_i \in \mathbb{R}$ considered between $\tau_{1i}$ and $\epsilon_{uui}$ i.e., $i^{th}$ element of $\epsilon_{uu}$.
In the subsequent analaysis, $\tau_1 \triangleq \tau_1(z)$.
\section{Variable gain gradient descent-based Integral Reinforcement Learning}\label{vargain}
\subsection{Integral reinforcement learning algorithm}
IRL can be considered as an alternate formulation of Bellman equation that does not require the drift dynamics.
It is obtained by integrating the infinitesimal version of (\ref{cost}) over time interval $[t-T,t]$ (where, $T$ is fixed time interval also known as reinforcement interval),
\begin{equation}
V(z_{t-T})=\int_{t-T}^te^{-\gamma(\tau-t+T)}[Q(z)+U(u)]d\tau+e^{-\gamma T}V(z_{t})
\label{irl1}
\end{equation}
In order to preserve the equivalence between (\ref{Ha}) and (\ref{irl1}), $T$ must be selected as small as possible \cite{modares2014optimal} and \cite{vrabie2009adaptive}.
Eq. (\ref{irl1}) is also known as IRL form of Bellman equation. Now, using the approximation from (\ref{appro}) in (\ref{irl1}), the approximation error can be expressed as:
\begin{equation}
\footnotesize
\begin{split}
\int_{t-T}^te^{-\gamma(\tau-t+T)}[Q(z)+U_1(u)]d\tau+&e^{-\gamma T}W^{T}\vartheta(z_t)
-W^{T}\vartheta(z_{t-T})\equiv \varepsilon_B
\end{split}
\label{ehjb}
\end{equation}
where, $U_1(u)$ is obtained by substituting (\ref{vstar}) in (\ref{hjb_int1}) i.e., $U_1(u)=-W^T\nabla\vartheta Gu+
u_m^2\sum_{i=1}^{m}R_i\log[1-\tanh^2{(\tau_{1i}(z)+\varepsilon_{u_i^*})}]$.
Now from \cite{modares2014optimal}, $\int_{t-T}^te^{-\gamma(\tau-t+T)}\dot{\vartheta}$ can be written as,
\begin{equation}
\footnotesize
\begin{split}
\int_{t-T}^te^{-\gamma(\tau-t+T)}\dot{\vartheta}=\int_{t-T}^te^{-\gamma(\tau-t+T)}\nabla{\vartheta}(F+Gu)=\Delta\vartheta+\gamma\int_{t-T}^te^{-\gamma(\tau-t+T)}\vartheta
\end{split}
\end{equation}

Then, from (\ref{ehjb}), $\Delta\vartheta(z_t)\triangleq e^{-\gamma T}\vartheta(z_t)-\vartheta(z_{t-T})$, where $\vartheta$ is the regressor vector for critic NN, can also be written as,
\begin{equation}
\begin{split}
\Delta\vartheta(z_t)=\int_{t-T}^te^{-\gamma(\tau-t+T)}[\nabla{\vartheta}(F+Gu)-\gamma\vartheta]d\tau
\end{split}
\label{delvarth}
\end{equation}
Now, substituting (\ref{delvarth}) and the expression for $U_1(u)$ obtained above, in (\ref{ehjb}) and upon simplification HJB approximation error can be re-written as,
\begin{equation}
\scriptsize
\begin{split}
\int_{t-T}^te^{-\gamma(\tau-t+T)}[z^TQ_1z-\gamma W^T\vartheta+W^T\nabla{\vartheta}F+u_m^2\sum_{i=1}^mR_i\log(1-\tanh^2{(\tau_{1i})})]d\tau=\varepsilon_{HJB}
\end{split}
\label{tracking_bell}
\end{equation}

Since ideal critic NN weights are not known, their estimates will be used instead. This results in approximate value as $\hat{V}(z)=\hat{W}^T\vartheta(z)$, where $\hat{W}$ is estimated weight. Thus, approximate optimal control $(\hat{u})$ and HJB error $(\hat{e})$ are then obtained as,
\begin{subequations}
\begin{equation}
\hat{u}=-u_m\tanh{\Big(\frac{1}{2u_m}R^{-1}G^T(z)\nabla{\vartheta}^T\hat{W}}\Big)
\label{approx_u}
\end{equation}
\vspace{-.1cm}
\begin{equation}
\small
\begin{split}
 \hat{e}&=\int_{t-T}^te^{-\gamma(\tau-t+T)}[Q(z)+\hat{U}(\hat{u})]d\tau+e^{-\gamma T}\hat{W}^{T}\vartheta(z_t)-\hat{W}^{T}\vartheta(z_{t-T})
\end{split}
\label{hjb_error}
\end{equation}
\end{subequations}
where, $\hat{U}\triangleq\hat{U}(\hat{u})$ is the estimated version of $U(u)$ obtained by substituting $\hat{V}_z=\nabla\vartheta^T\hat{W}$ in $A(z)$ in (\ref{hjb_int1}) and given as,
\begin{equation}
\small
\begin{split}
\hat{U}=2u_m^2\tau_{2}^T(z)R\tanh{\tau_2(z)}+
u_m^2\sum_{i=1}^{m}R_i\log[1-\tanh^2{\tau_{2i}(z)}]
\end{split}
\label{Uhat}
\end{equation}
 where, $\tau_2=(1/2u_m)R^{-1}G(z)^T \nabla{\vartheta}^T\hat{W} \in \mathbb{R}^m$.
Now using (\ref{delvarth}) and subtracting (\ref{tracking_bell}) from (\ref{hjb_error}), the HJB error can be expressed in terms of $\tilde{W}$ as obtained in \cite{modares2014optimal}, 
\begin{equation}
\begin{split}
\hat{e}=-\Delta\vartheta^T\tilde{W}+\int_{t-T}^te^{-\gamma(\tau-t+T)}\tilde{W}^TMd\tau+E
\end{split}
\end{equation}
\begin{equation}
\small
\begin{split}
M&=\nabla{\vartheta}G(z)u_m\mathcal{F}(z),~\mathcal{F}(z)=\tanh{(\tau_2(z))}-sgn(\tau_2(z))\\
E&=\int_{t-T}^te^{-\gamma(\tau-t+T)}\Big(W^T\nabla{\vartheta}Gu_m(sgn(\tau_{2})-sgn(\tau_{1}))\\
&+u_m^2R(\varepsilon_{\tau_2}-\varepsilon_{\tau_1})+\varepsilon_{HJB}\Big)d\tau
\end{split}
\label{em}
\end{equation}
\subsection{Online variable gain gradient descent-based update law}
In \cite{modares2014optimal} and \cite{vamvoudakis2014online}, to tune the critic NN weights in order to minimize the approximate HJB error $\hat{e}$, a gradient descent-based online update scheme was used as, 
\begin{equation}
\begin{split}
\dot{\hat{W}}=-\alpha\bar{\vartheta}\hat{e}=-\alpha\bar{\vartheta}(\hat{I}+\hat{W}^T\Delta\vartheta)
\end{split}
\label{vamoda}
\end{equation}
where $\alpha>0$ is the constant learning rate, $\hat{I}$ denotes the reinforcement integral $(\hat{I}=\int_{t-T}^te^{-\gamma(\tau-t+T)}[Q(z)+\hat{U}(\hat{u})]d\tau)$ and $\bar{\vartheta}$, normalized regressor is given as, $\bar{\vartheta}=\frac{\Delta{\vartheta}}{(1+\Delta{\vartheta}^T\Delta{\vartheta})^2}$.

In order to update the critic NN weights in IRL framework, a novel online parameter update law that is driven by variable gain gradient descent is presented below. 
\begin{equation}
\small
\begin{split}
\dot{\hat{W}}&=-\alpha|\hat{e}|^{q_2}\bar{\vartheta}\hat{e}+\frac{\alpha}{2}|\Sigma|^{k_2}\Xi(z,\hat{u})\nabla{\vartheta}G(z)[I_m-\mathcal{B}(\tau_{2}(z))]G^T(z)L_{2z}\\
&+\alpha|\hat{e}|^{q_2}\Big((K_1\varphi^T-K_2)\hat{W}-\bar{\vartheta}\int_{t-T}^te^{-\gamma(\tau-t+T)}\hat{W}^TMd\tau \Big)
\end{split}
\label{tuning_law1}
\end{equation}
In (\ref{tuning_law1}), $\alpha>0,~k_2>0,~q_2>0$ and $K_1 \in \mathbb{R}^{N_1}~ K_2 \in \mathbb{R}^{N_1\times N_1}$ are constants.
The rate of variation of Lyapunov function along the augmented system trajectories is given by, $\Sigma=L^T_{2z}(F(z)+G(z)\hat{u})=z^T\dot{z}$, where $L_{2}=(1/2z^Tz)$.
Note, that since, augmented drift dynamics, $F(z)$ is unknown, the term $\Sigma$ is computed using numerical differentiation of the augmented state vector $z$, i.e., $\Sigma=z^T\dot{z}$.
Other terms used in the update law above are, $\varphi=\Delta\vartheta/m_s$, $m_s=1+\Delta\vartheta^T\Delta\vartheta$, $\mathcal{B}=diag\{\tanh^2{(\tau_{2i}(z))}\},~i=1,2...,m$, and $M$ is as in (\ref{em}).
The term $\Xi(z,\hat{u})$ is a piece-wise continuous indicator function defined as,
\vspace{-.1cm}
\begin{equation}
\Xi(z,\hat{u})=  \begin{cases} 
      0, & if~ \Sigma <0  \\
      1, & otherwise
   \end{cases}
\end{equation}
The update law presented in (\ref{tuning_law1}) is different from the IRL update laws mentioned in \cite{vrabie2009adaptive}, \cite{vamvoudakis2014online} and  \cite{modares2014optimal} in many ways. 
Both \cite{modares2014optimal} and \cite{vamvoudakis2014online} use dual-approximation structure known as actor-critic to compute optimal cost and optimal control and require an initial stabilizing controller to initiate the process of policy iteration, while the update law presented here leverages only one approximator (critic) and does not require any initial stabilizing controller. 
Rather it is an expansion over the update law presented in \cite{liu2015reinforcement} to IRL framework where the knowledge of drift dynamics is not required and a variable learning rate gradient descent is used.
The novelty of the tuning law (\ref{tuning_law1}) is explained below.
\begin{itemize}
\item The first term in (\ref{tuning_law1}) reduces the HJB error. 
It is different from the existing gradient-based update laws in the sense that it includes variable learning rate via the term $\alpha|\hat{e}|^{q_2}$. 
This term scales the learning speed based on the instantaneous value of the HJB error, i.e., a higher value of the HJB error in the initial phases of learning leads to a higher learning speed, while as the HJB error decreases, the learning speed is also reduced.
\item A stabilizing term comes next in (\ref{tuning_law1}). An indicator function similar to \cite{dierks2010optimal} and \cite{liu2015reinforcement} is considered so that this term is zero when the Lypunov function is strictly decreasing along the system trajectories and becomes one when the Lyapunov function is non-decreasing along the system trajectories. 
The main deviation from \cite{liu2015reinforcement} lies in the use of variable learning rate $\alpha|\Sigma|^{k_2}$ instead of just the constant $\alpha$. 
This augmentation could potentially pull the system out of unstable region towards stable region at a rate proportional to the degree of instability.
\item The last term in (\ref{tuning_law1}) is meant for enhancement of robustness. A similar concept of robust term was also presented in \cite{liu2015reinforcement}. However, in (\ref{tuning_law1}) the robust term appears with a variable learning rate and reflect the changes due to IRL framework.
\item Compared to \cite{liu2015reinforcement}, this update law does not require the information of drift dynamics and compared to \cite{modares2014optimal} and \cite{vamvoudakis2014online} it does not require actor NN and an initial stabilizing controller.
\end{itemize}

The error in critic weight is defined as, $\tilde{W}=W-\hat{W}$. The critic weight error dynamics is then given as,
\begin{equation}
\scriptsize
\begin{split}
\dot{\tilde{W}}&=\alpha(|e|^{q_2})\bar{\vartheta}[-\Delta\vartheta^T\tilde{W}+\int_{t-T}^te^{-\gamma(\tau-t+T)}\tilde{W}Md\tau+E]-\frac{\alpha}{2}(|\Sigma|^{k_2})\Xi(z,\hat{u})\nabla{\vartheta}G(z)\Big[I_m\\
&-\mathcal{B}(z)\Big]G^T(z)L_{2z} +\alpha(|e|^{q_2})\Big[\bar{\vartheta}\int_{t-T}^te^{-\gamma(\tau-t+T)}\hat{W}Md\tau+(K_2-K_1\varphi^T)\hat{W}\Big]
\end{split}
\label{crit_err}
\end{equation}
\vspace{-.5cm}
\subsection{Stability proof of the online update law}
\begin{assumption}\label{weight_as}
\textnormal{Ideal NN weight vector $W$ is considered to be bounded by a positive constant $W_M>0$ such that $\|W\| \leq  W_M$. There exists positive constants $b_{\epsilon}$ and $b_{\epsilon z}$ that bound the approximation error and its gradient such that $\|\varepsilon(z)\| \leq b_{\epsilon}$ and $\|\nabla{\varepsilon}\| \leq b_{\epsilon z}$.} 
\end{assumption}

\begin{assumption}\label{regres_as}
\textnormal{Critic regressors are considered to be bounded as well: $\|\vartheta(z)\| \leq b_{\vartheta}$ and $\|\nabla{\vartheta}(z)\| \leq b_{\vartheta z}$. }
\end{assumption}
In this paper, both the assumptions hold true because, there exists a stabilizing term in the update law (\ref{tuning_law1}) that comes into effect when Lyapunov function starts growing along the system trajectories. This term helps in pulling the system out of instability ensuring that the trajectories remain finite within a region $z \in \Omega \subset R^{2n}$. 
The Assumptions, \ref{weight_as} and \ref{regres_as} are in line with assumptions made in \cite{liu2015reinforcement}.

\begin{assumption}\label{lyapunov_as}
\textnormal{Let $L_2 \in C^1$ be a continuously differentiable and radially unbounded Lyapunov candidate for (\ref{aug}) and satisfies $\dot{L}_2=L^T_{2z}(F(z)+G(z)u^*)<0$, where $u*$ is as defined in (\ref{opt_u}). Furthermore, there exists a symmetric and positive definite $\Lambda \in \mathbb{R}^{2n\times 2n}$ such that $L^T_{2z}(F(z)+G(z)u^*)=-L^T_{2z}\Lambda L_{2z}$, where $L_{2z}$ is the partial derivative of $L_{2}$ wrt $z$.}
\end{assumption}
Following Lipschitz continuity of $(F(z)+Gu^*)$ in $z$, this assumption can be shown reasonable. It is also in line with Eq. (7) in \cite{dierks2010optimal} and Assumption 4 of \cite{liu2015reinforcement}.

\begin{theorem}
Let the CT nonlinear augmented system be described by (\ref{aug}) with associated HJB as (\ref{hjb_int}) and approximate optimal control as (\ref{approx_u}), then the tuning law (\ref{tuning_law1}) is stable in the sense of UUB.
\end{theorem}

\begin{proof}
Let the Lyapunov candidate be 
\begin{equation}
    L=L_2+(1/2\alpha)\tilde{W}^T\tilde{W}
\end{equation}
where, $L_2$ is as described after (\ref{tuning_law1}). 
In subsequent analysis, the terms $|\hat{e}(t)|^{k_2}$ ($\hat{e}(t)$ as obtained from (\ref{hjb_error})) and $|\Sigma(z(t))|^{q_2}$ would be referred to as $g_1(t)$ and $g_2(t)$, respectively. 
Then,
\begin{equation}
\small
\begin{split}
\dot{L}=L^T_{2z}(F(z)-u_m G(z)\tanh{(\tau_2(z))})+\tilde{W}^T\dot{\tilde{W}}/\alpha
\end{split}
\end{equation}
Utilizing error dynamics of weights (\ref{crit_err}) and the fact that $\dot{z}=F(z)+G(z)\hat{u}$, the last term of $\dot{L}$ becomes:
\begin{equation}
\small
\begin{split}
&\tilde{W}^T\alpha^{-1}\dot{\tilde{W}}=g_1\tilde{W}^T\bar{\vartheta}\Big[-\Delta\vartheta^T\tilde{W}+\int_{t-T}^te^{-\gamma(\tau-t+T)}\tilde{W}Md\tau+E\Big]-\\
& \frac{1}{2}g\Xi(z,\hat{u})L^T_{2z}G(z)\Big[I_m-\mathcal{B}(\tau_2(z))\Big]G^T(z)\nabla{\vartheta}^T\tilde{W} \\
&+g_1\tilde{W}^T\bar{\vartheta}\int_{t-T}^te^{-\gamma(\tau-t+T)}\hat{W}Md\tau+g_1\tilde{W}^T(K_2\hat{W}-K_1\varphi^T\hat{W}) \\
&=-g_1\tilde{W}\varphi\varphi^T\tilde{W}+g_1\tilde{W}^T(\varphi^T/m_s)E+g_1\tilde{W}^T\beta(z)-\\
& \frac{1}{2}g_2\Xi(z,\hat{u})L^T_{2z}G[I_m-\mathcal{B}(\tau_2(z))]G^T(z)\nabla{\vartheta}^T\tilde{W}+g_1\tilde{W}^T(K_2\hat{W}\\
&-K_1\varphi^T\hat{W})
\end{split}
\label{ldot4}
\end{equation}
where, $\bar{\vartheta}=\varphi/m_s$, $\varphi=\Delta\vartheta/m_s$, $\beta(z)=(\varphi^T/m_s)I_0$ where $I_0=\int_{t-T}^te^{-\gamma(\tau-t+T)}W^TMd\tau$.
Now, define 
\begin{equation}
\begin{split}
  \mathcal{S}\triangleq[\tilde{W}^T\varphi,\tilde{W}^T]^T  
\end{split}
\label{j_def}
\end{equation}
Then, (\ref{ldot4}) can be re-written as:
\begin{equation}
\begin{split}
\dot{\tilde{W}}^T\tilde{W}/\alpha&=-\mathcal{S}^TM\mathcal{S}+\mathcal{S}^T\mathcal{K}-\frac{1}{2}g_2\Xi(z,\hat{u})L^T_{2z}G(z)[I_m\\
&-\mathcal{B}(\tau_2(z))]G^T(z)\nabla{\vartheta}^T\tilde{W}
\end{split}
\end{equation}
where, $g_1 \triangleq g_1(t), g_2\triangleq g_2(t)$, $\mathcal{K}=g_1N\in \mathbb{R}^{N_1+1}$ and $M\in \mathbb{R}^{(N_1+1)\times (N_1+1)}$ are defined as,
\begin{equation}
\small
\begin{split}
M=\begin{pmatrix}
g_1    &      -\frac{g_1}{2}K^T_1 \\
-\frac{g_1}{2}K_1  &  g_1K_2
\end{pmatrix},N=\begin{pmatrix}
E/m_s  \\
\beta(z)+K_2W_c-K_1\varphi^TW
\end{pmatrix}
\end{split}
\end{equation}
Therefore, the Lyapunov derivative can be rendered as,
\begin{equation}
\small
\begin{split}
\dot{L} &\leq L^T_{2z}(F(z)+G(z)\hat{u})-\lambda_{min}(M_{1})\|\mathcal{S}\|^2+b_N\|\mathcal{S}\|\\
& -\frac{1}{2}g_2\Xi(z,\hat{u})L_{2z}^TG(z)[I_m-\mathcal{B}(\tau_2(z))]G^T(z)\nabla{\vartheta}^T\tilde{W}
\end{split}
\label{ldot}
\end{equation}
where, $M_{1}=(M+M^T)/2$. 
The upper bound of $\mathcal{K}$, is given as, $\|\mathcal{K}\| \leq b_N=g_1max(\|N\|)$.
Based on the value of the piece-wise continuous function, $\Xi(z,\hat{u})$, Eq. (\ref{ldot}) can be explained in two cases as detailed below.

\textbf{Case (i)}:
${\Xi(z,\hat{u})=0}$.\\
By definition, in this case, $L_{2z}^T\dot{z}<0$. Then, considering the density property of real numbers it can be said that there exists  $\beta$ such that $0<\beta \le \|\dot{z}\|$. Therefore, $L_{2z}^T\dot{z} \leq -\|L_{2z}\|\beta <0$.
Now, using this inequality in (\ref{ldot}), $\dot{L}$ is obtained as,
\begin{subequations}
\begin{equation}
\small
 \dot{L}  \leq L^T_{2z}\dot{z}-\lambda_{min}(M_1)\|\mathcal{S}\|^2+b_N\|\mathcal{S}\|   
\label{ldot2}
\end{equation}
\begin{equation}
\small
\leq -\|L_{2z}\|\beta+\frac{b_N^2}{4\lambda_{min}(M_1)}-\lambda_{min}(M_1)\Big(\|\mathcal{S}\| -\frac{b_N}{2\lambda_{min}(M_1)}\Big)^2 
\label{ldot3}
\end{equation}
\end{subequations}
Form (\ref{ldot2}), a sufficient condition to ensure negative definiteness of $\dot{L}$ for the above system can be obtained as $\|\mathcal{S}\| > b_N/\lambda_{min}(M_1)$. Also, recall from the definition of $\mathcal{S}$ in (\ref{j_def}), the upper bound of $\|\mathcal{S}\|$ can be obtained as, $ \|\mathcal{S}\| \leq \Big(\sqrt{1+\|\varphi\|^2}\Big)\|\tilde{W}\|$.
Therefore, using the upper and lower bound over $\|\mathcal{S}\|$ as obtained above and the expression for $b_N$, the UUB set for $\tilde{W}$ can be obtained as, 
\begin{equation}
\small
\|\tilde{W}\| > \frac{g_1max(\|N\|)}{\lambda_{min}(M_1)\sqrt{1+\|\varphi\|^2}}
\label{ld1}
\end{equation}
Now from, (\ref{ldot3}), for ensuring stability,
\begin{equation}
\small
\begin{split}
-\|L_{2z}\|\beta+\frac{b_N^2}{4\lambda_{min}(M)} < 0
\end{split}
\label{l2_case1}
\end{equation}
Utilizing $b_N$ in (\ref{l2_case1}), the UUB set for $L_{2z}$ can be obtained as,
\begin{equation}
\|L_{2z}\| > \frac{g_1^2max(\|N\|^2)}{4\beta\lambda_{min}(M_1)} 
\label{declyap}
\end{equation}
This proves that $\tilde{W}$ and  $L_{2z}$ are UUB stable with corresponding sets described by (\ref{ld1}) and (\ref{declyap}), respectively. 

\textbf{Case (ii)}: ${\Xi(z,\hat{u})=1}$ \\
This case implies that the control policy generated during policy iteration is destabilizing.
From (\ref{approx_u}) and (\ref{ldot}),
\begin{equation}
\small
\begin{split}
\dot{L} & \leq L^T_{2z}F(z)-u_m L^T_{2z}G(z)\Big[\tanh{(\tau_2(z))}+\frac{g_2}{2u_m}[I_m\\
&-\mathcal{B}(\tau_2(z))]G^T\nabla{\vartheta}^T\tilde{W}\Big]-\lambda_{min}(M_1)\|\mathcal{S}\|^2+b_N\|\mathcal{S}\| \\
\end{split}
\end{equation}
Now, adding and subtracting $L^T_{2z}(G(z)u^*)$,
\begin{equation}
\small
\begin{split}
\dot{L} & \leq L^T_{2z}(F(z)+Gu^*)-u_m L^T_{2z}G(z)\Big[\tanh{(\tau_2(z))}\\
& +\frac{g_2}{2u_m}[I_m-\mathcal{B}(\tau_2(z))]G^T\nabla{\vartheta}^T\tilde{W}\Big]-\lambda_{min}(M_1)\|\mathcal{S}\|^2\\
& +b_N\|\mathcal{S}\|-L^T_{2z}G(z)(-u_m\tanh{(\tau_1(z))}+\epsilon_{u^*})
\end{split}
\label{ldot_2case}
\end{equation}
Now, using the inequality $\|\tanh{(\tau_1(z))}-\tanh{(\tau_2(z))}\| \leq 2\sqrt{m}$ (where, $m$ is the dimension of control vector), Assumptions \ref{regres_as} and \ref{lyapunov_as},  Inequality (\ref{ldot_2case}) can be re-written as:
\begin{equation}
\footnotesize
\begin{split}
\dot{L} 
& \leq -\lambda_{min}(\Lambda)\|L_{2z}\|^2+\|L_{2z}\|(2\sqrt{m}u_m g_M+\frac{g_2}{2}\|\mathcal{N}_1\nabla^T{\vartheta}\tilde{W}\|)+\frac{b_N^2}{4\lambda_{min}(M_1)}\\
&-\lambda_{min}(M_1)\Big(\|\mathcal{S}\|-\frac{b_N}{\lambda_{min}(M_1)}\Big)^2+k\|L^T_{2z}\|g_M^2b_{\vartheta z}\\
\end{split}
\label{first_ineq}
\end{equation}
where, $\mathcal{N}_1\triangleq G(z)[\mathcal{B}(\tau_2(z))-I_m]G^T(z)$, $k \triangleq\frac{1}{2}(1-\tanh^2{\xi})$. Now, following similar method as in \cite{liu2015reinforcement}, two positive constant numbers $l_1$ and $l_2$ are defined such that $l_1+l_2=1$. 
From (\ref{j_def}), $\|\tilde{W}\|^2 \leq \|\mathcal{S}\|^2$, the inequality in (\ref{first_ineq}) can be developed as follows:
\begin{equation}
\scriptsize
\begin{split}
\dot{L}
\leq -l_1\lambda_{min}(\Lambda)(\|L^T_{2z}\|-\mathcal{Q}_2)^2+\mathcal{Q}_1-(\lambda_{min}(M_1)-\mathcal{Q})\Big(\|\mathcal{S}\|-\frac{b_N}{2(\lambda_{min}(M_1)-\mathcal{Q})}\Big)^2\\
\end{split} 
\label{ldot_semi_final}
\end{equation}
where, 
\begin{equation}
\small
\begin{split}
\mathcal{Q}&\triangleq\frac{\|g_2/2\mathcal{N}_1\nabla^T\vartheta\|^2}{4l_2\lambda_{min}(\Lambda)};~~\mathcal{Q}_2 \triangleq \frac{2\sqrt{m}g_Mu_m+kg^2_Mb_{\vartheta z}}{2l_1\lambda_{min}(\Lambda)} \\
\mathcal{Q}_1& \triangleq\frac{(2\sqrt{m}g_Mu_m+kg^2_Mb_{\vartheta z})^2}{4l_1\lambda_{min}(\Lambda)}+\frac{b_N^2}{4(\lambda_{min}(M_1)-\mathcal{Q})}
\end{split}
\label{R}
\end{equation}
Note that $\mathcal{Q}$ must be selected such that it is smaller than $\lambda_{min}(M_1)$. 
 Finally, inequality of (\ref{ldot_semi_final}) can be split into two inequalities, one for $L^T_{2z}$ and the other one for $\mathcal{S}$ (or equivalently $\tilde{W}$) as follows.
 \begin{subequations}
 \begin{equation}
\begin{split}
-l_1\lambda_{min}(\Lambda)(\|L^T_{2z}\|-\mathcal{Q}_2)^2+\mathcal{Q}_1<0
\end{split}
\label{L2}
\end{equation}
\vspace{-.15cm}
\begin{equation}
\small
\begin{split}
\mathcal{Q}_1-(\lambda_{min}(M)-\mathcal{Q})\Big(\|\mathcal{S}\|-\frac{b_N}{2(\lambda_{min}(M_1)-\mathcal{Q})}\Big)^2 <0 
\end{split}
\label{Y}
\end{equation}
 \end{subequations}
Therefore, (\ref{L2}) leads to UUB set for $L_{2z}$,
\vspace{-.1cm}
\begin{equation}
\small
\begin{split}
\|L^T_{2z}\| > \mathcal{Q}_2+\sqrt{\frac{\mathcal{Q}_1}{l_1\lambda_{min}(\Lambda)}}
\end{split}
\label{L2c2}
\end{equation}

Similarly, for $\mathcal{S}$, (\ref{Y}) can be re-written as,
\begin{equation}
\small
\begin{split}
\|\mathcal{S}\| > \frac{b_N}{2(\lambda_{min}(M_1)-\mathcal{Q})}+\sqrt{\frac{\mathcal{Q}_1}{(\lambda_{min}(M_1)-\mathcal{Q})}}=A 
\label{Yineq1}
\end{split}
\end{equation}
Utilizing the upper bound of $\|\mathcal{S}\|$ defined after (\ref{ldot3}) and lower bound of $\|\mathcal{S}\|$ from (\ref{Yineq1}), UUB set for $\tilde{W}$ is,
\vspace{-.1cm}
\begin{equation}
\begin{split}
\|\tilde{W}\|>\frac{A}{\sqrt{1+\|\varphi\|^2}}
\end{split}
\label{wcineq}
\end{equation}
This completes the stability proof of the update mechanism (\ref{tuning_law1}).
\end{proof}

Note that for \textbf{Case (i)}, if $\tilde{W}$ and $L_{2z}$ stays outside (\ref{ld1}) and (\ref{declyap}), respectively,  and for \textbf{Case (ii)}, if they stay outside (\ref{wcineq}) and (\ref{L2c2}), respectively, then it leads to decreasing $\tilde{W}$ and $L_{2z}$.
A diminishing $\tilde{W}$ implies that $\hat{W}$ is getting close to the ideal weight $W$ which in turn implies a decreasing HJB error or a decreasing $g_1$.
Similarly, based on the expression of $L_{2}$ (defined after Eq. (\ref{tuning_law1})), and invoking Lipschitz continuity on dynamics, $|\Sigma|$ is bounded by $(L_F\|z\|+g_Mu_m)\|L_{2z}\|$ where, $\|F(z)\| \leq L_F\|z\|$ and $\|G(z)\hat{u}\| \leq g_Mu_m$, where $L_F$ and $g_M$ are Lipschitz constants on $F$ and $G$, respectively.
Now, since, $\|L_{2z}\|=\|z\|$ decreases in the stable region, it leads to $g_2$ being smaller than $(L_F\|z\|+g_Mu_m)\|z\|$.
As an immediate consequence of these, it can be observed that, the RHS of inequalities, (\ref{ld1}), (\ref{declyap}), (\ref{L2c2}) and (\ref{wcineq}) shrink in magnitude due to presence of terms $g_1$ and $g_2$.
This ensures that the variable gain gradient descent leads to tighter residual sets for parameters and augmented system trajectories than constant learning rate-based gradient descent scheme existing in literature.
\section{Result and Simulation}\label{res}
\vspace{-.1cm}
The control algorithm based on IRL and single approximation structure developed in this paper is validated in this section on the 6-DoF nonlinear model of Aerosonde UAV (refer to Pages 61,62 and 276 of \cite{beard2012small}). 
\begin{figure}
\centering
\subcaptionbox{Desired and actual attitude profile\label{fig:att}}{\includegraphics[width=.23\textwidth,height=9.5cm,keepaspectratio,trim={1.8cm 0.0cm 4cm .08cm},clip]{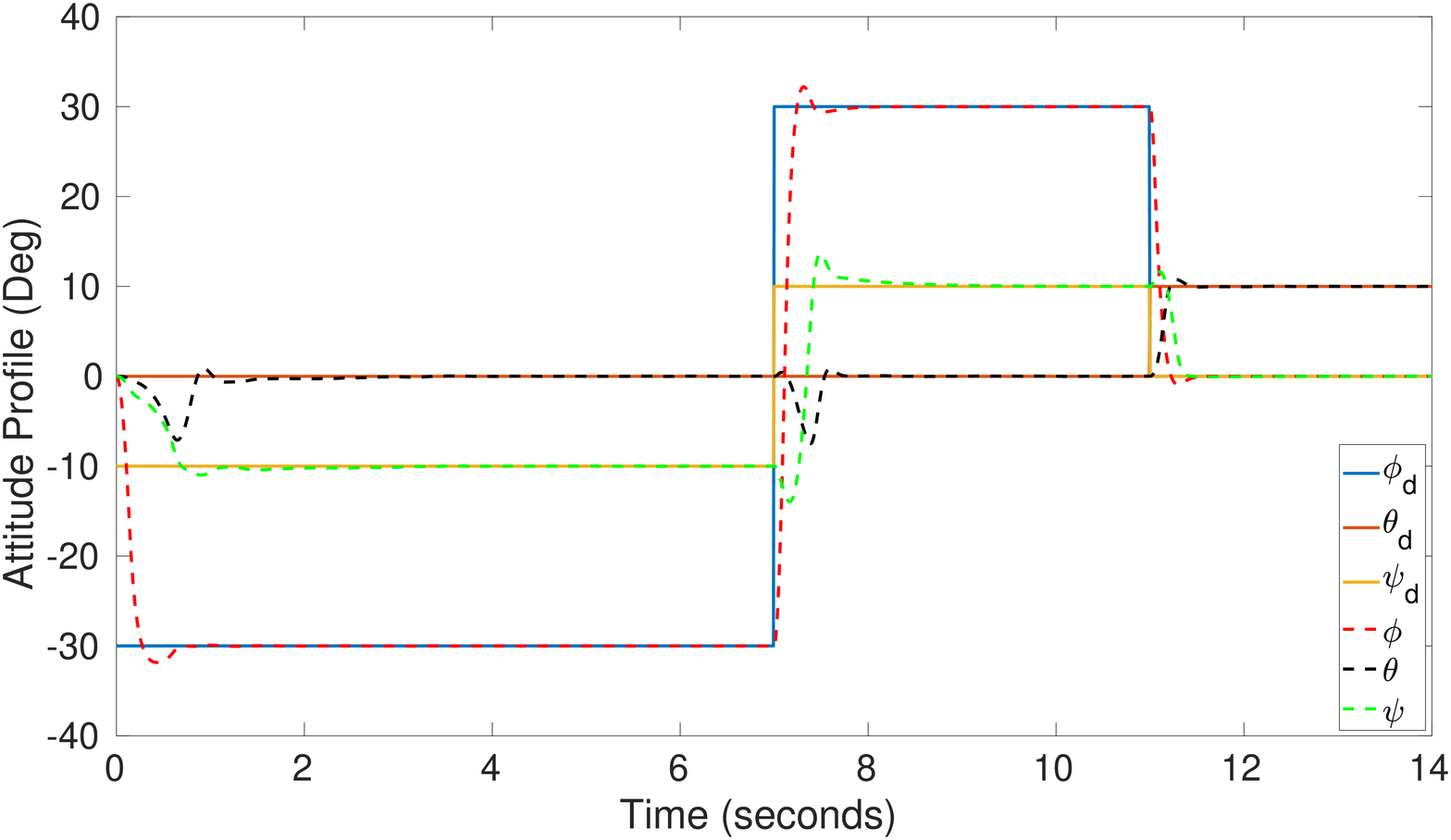}}%
\hspace{0cm} 
\subcaptionbox{Difference between desired and actual rates\label{fig:rates}}{\includegraphics[width=.23\textwidth,height=9.5cm,keepaspectratio,trim={1.8cm 0.0cm 2cm .08cm},clip]{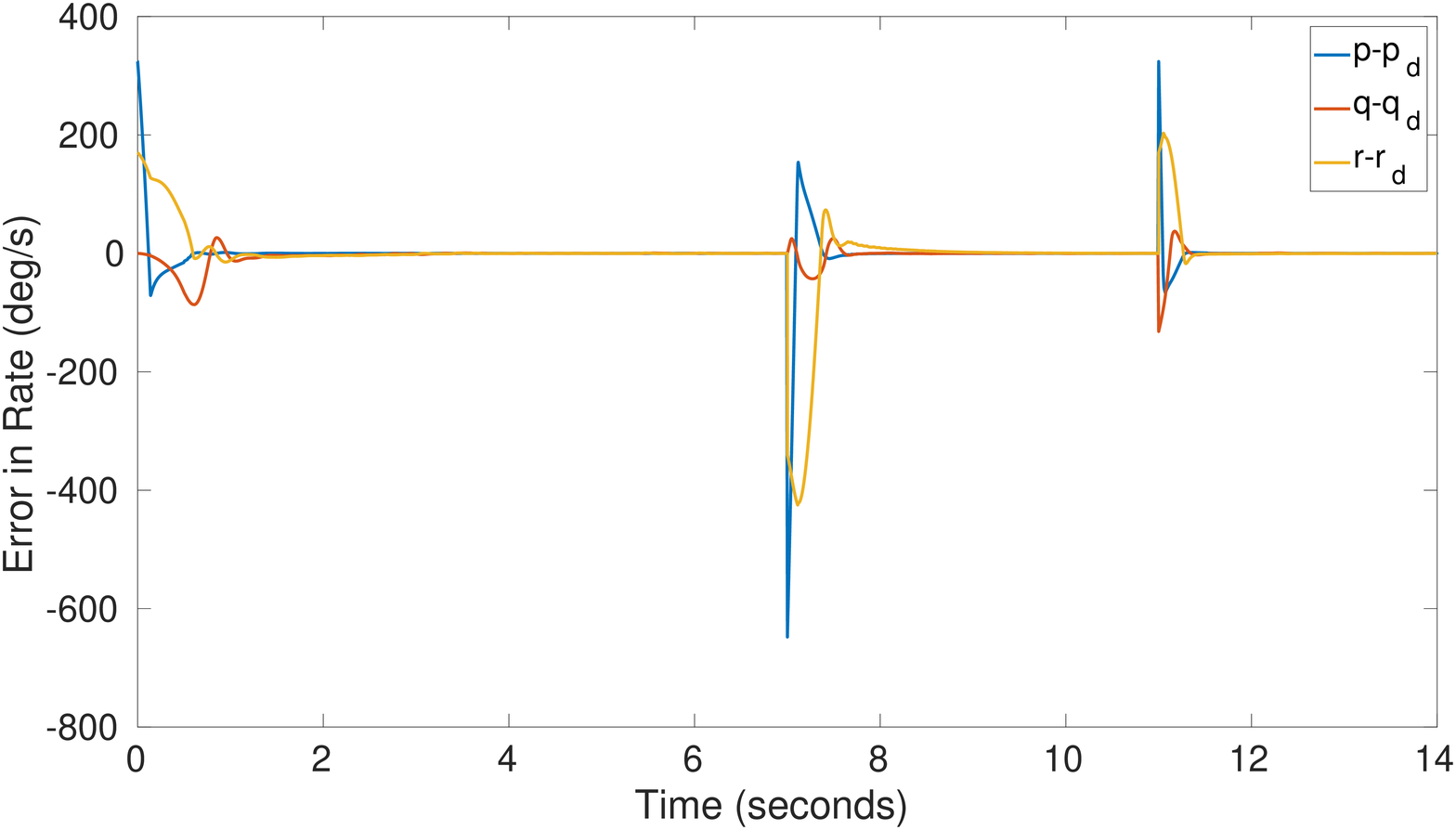}}\\

\hspace{0cm} 
\subcaptionbox{States profile of the UAV\label{fig:states}}{\includegraphics[width=.23\textwidth,height=9.5cm,keepaspectratio,trim={1.8cm 0.0cm 4cm .08cm},clip]{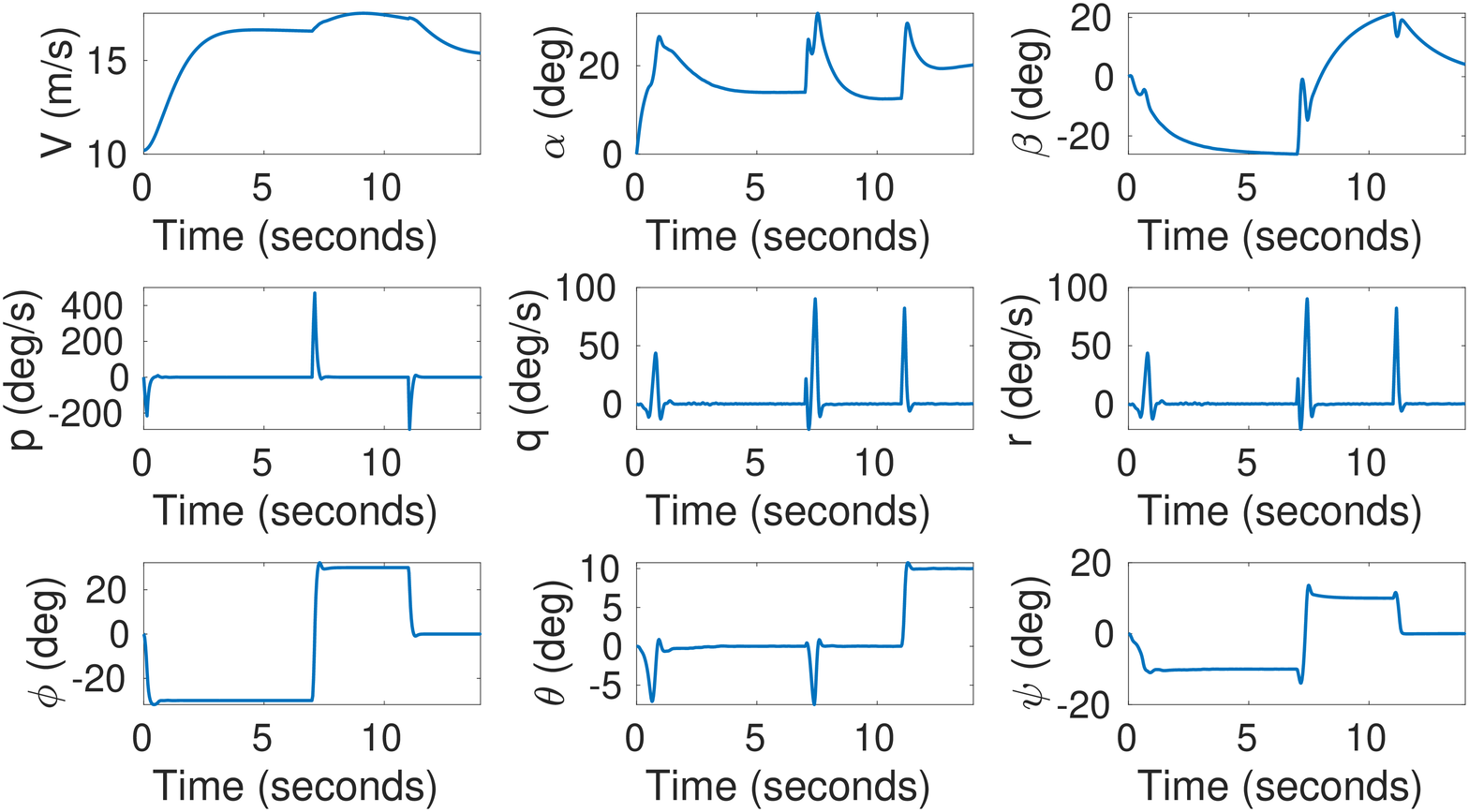}}%
\hspace{0cm} 
\subcaptionbox{Control profile representing elevator, eileron and rudder\label{fig:ctr}}{\includegraphics[width=.23\textwidth,height=9.5cm,keepaspectratio,trim={1.8cm 0.0cm 2cm .08cm},clip]{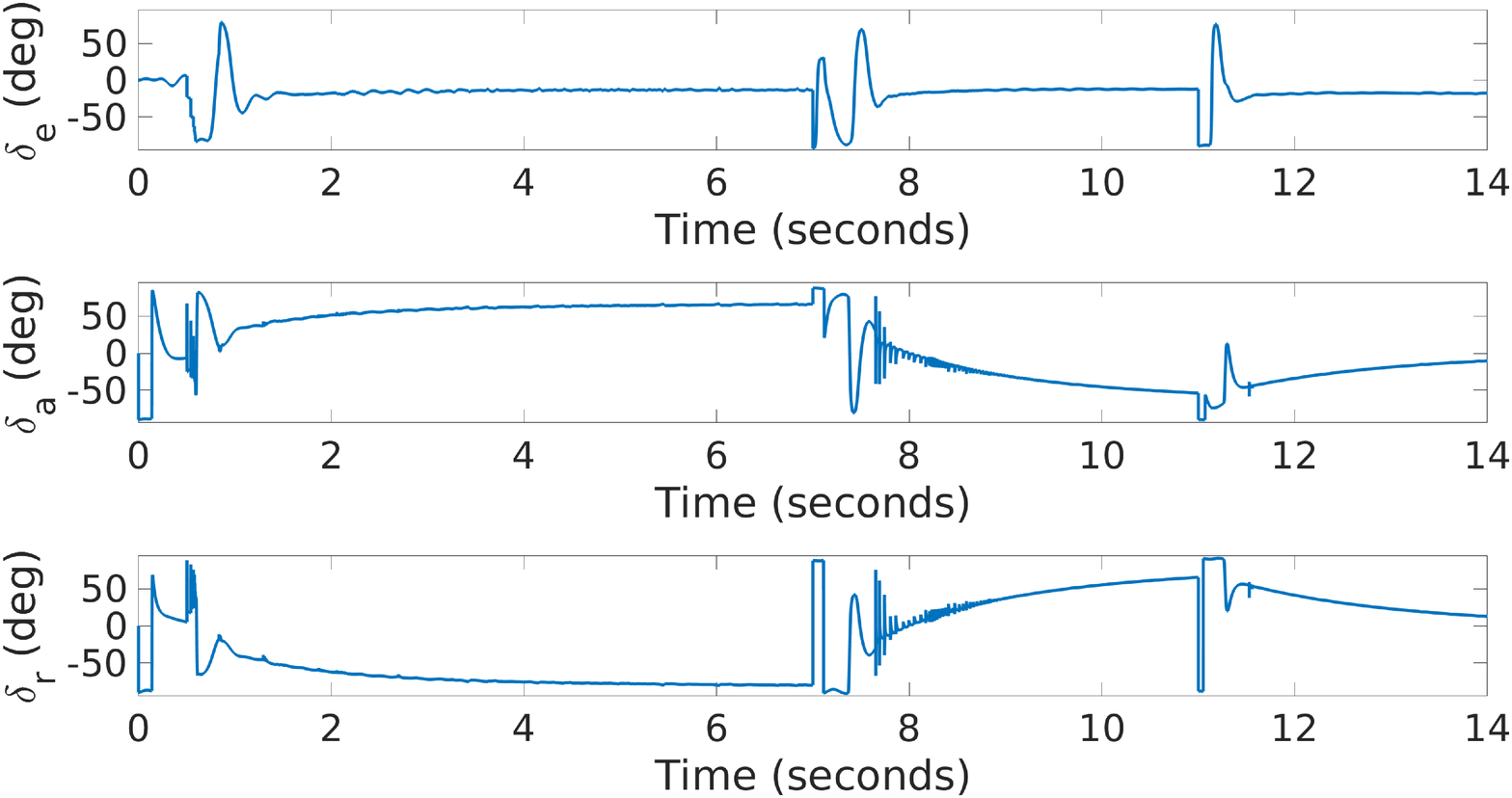}}%
\\
\hspace{0cm}
\subcaptionbox{Online evolution of critic NN weights\label{fig:crit}}{\includegraphics[width=.23\textwidth,height=9.5cm,keepaspectratio,trim={1.8cm 0.0cm 2cm .08cm},clip]{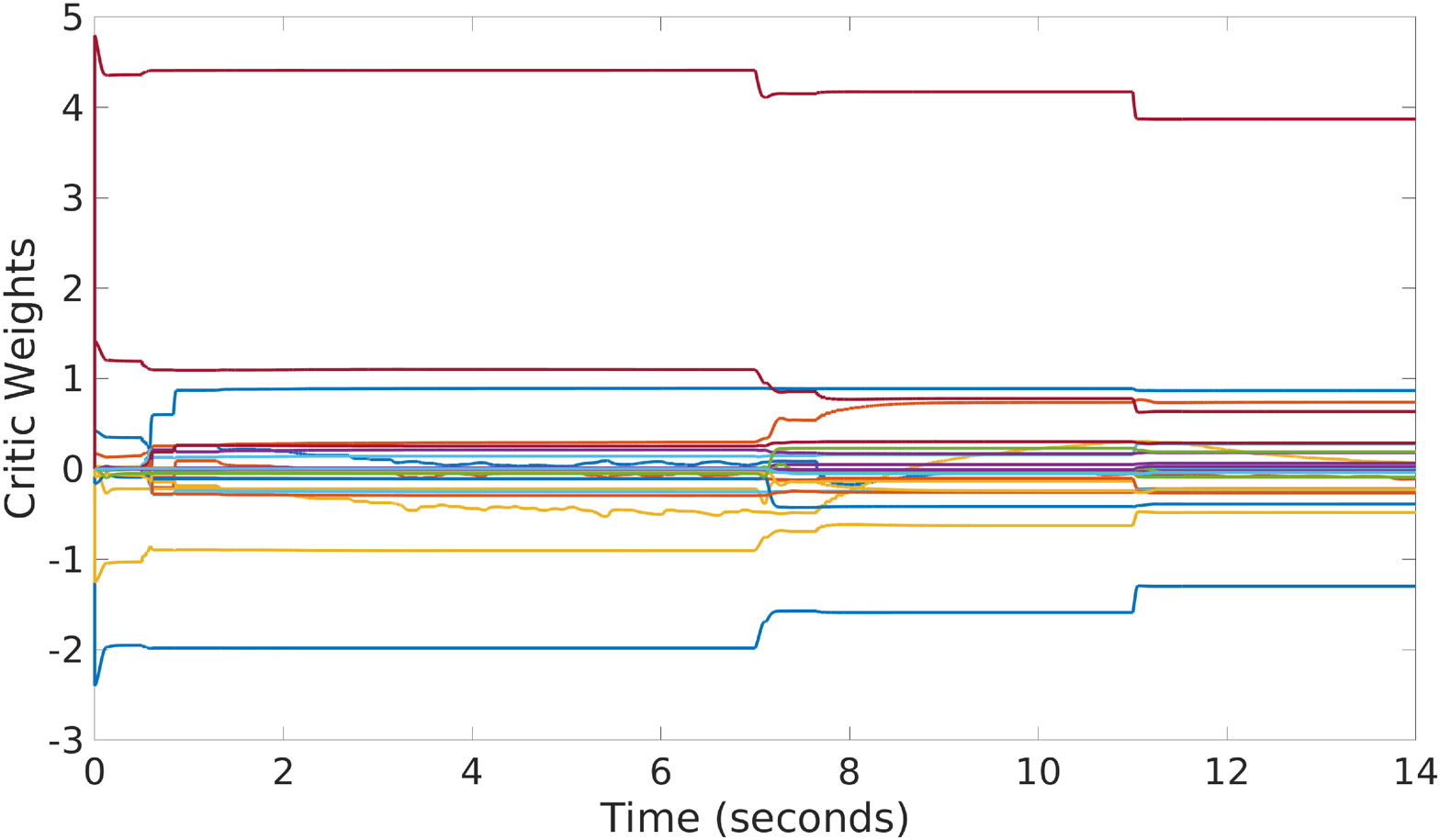}}
\hspace{0cm} 
\subcaptionbox{Value function approximation\label{fig:val}}{\includegraphics[width=.23\textwidth,height=9.5cm,keepaspectratio,trim={1.8cm 0.0cm 2cm .08cm},clip]{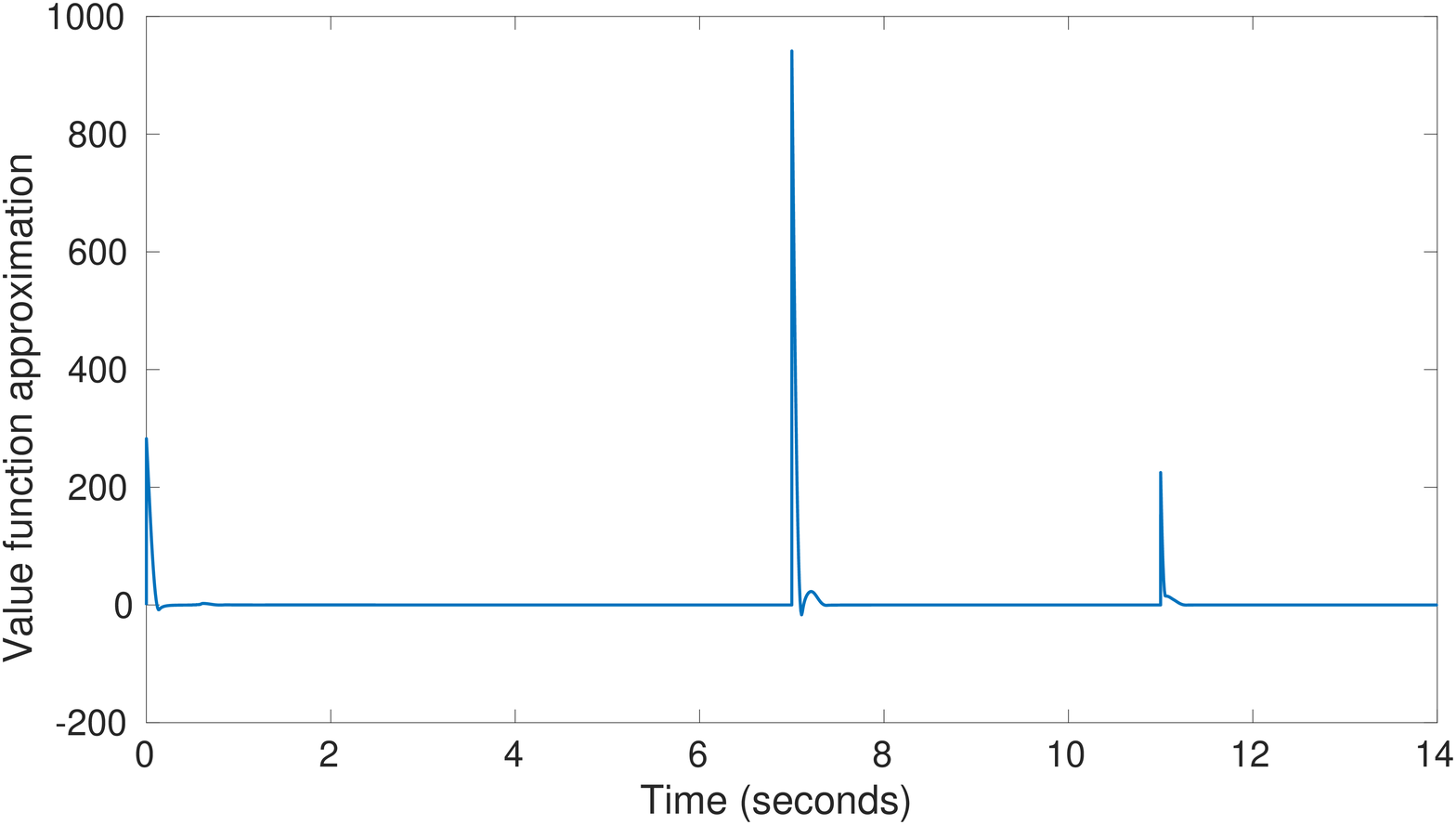}}%
\caption{Performance of Aerosonde UAV under presented update law}
\label{fig:e_states}
\end{figure}
The control implementation is made up of two cascaded loops, the the first loop, i.e., outer loop converts the desired Euler angle information to desired rates, the inner loop uses the presented control algorithm to track the desired rates in an optimal way. 
Desired Euler angle rates are given by, $p_{des}=\dot{\phi}_{des}-8 e_{\phi},~q_{des}=\dot{\theta}_{des}-10 e_{\theta},~r_{des}=\dot{\psi}_{des}-12 e_{\psi}$, where $\phi,\theta,\psi$ are roll, pitch and yaw angles, respectively. 
Elevator ($\delta_e$), Aileron ($\delta_a$) and Rudder ($\delta_r$) deflection are the control inputs to the system.
The augmented state is, $z=[e_{p},e_{q},e_{r},p_{des},q_{des},r_{des}]^T \in \mathbb{R}^6$ where $e=x-x_{des}$. 
The regressor vector for critic NN is chosen to be, $\vartheta=[z_1,z_2,z_3,z_4,z_5,z_6,z_1^2,z_2^2,z_3^2,z_4^2,z_5^2,z_6^2,z_1z_2,\\z_1z_3,z_1z_4,z_1z_5,z_1z_6,z_2z_3,z_2z_4,z_2z_5,z_2z_6,z_3z_4,z_3z_5,z_3z_6,z_4z_5,\\z_4z_6,z_5z_6]^T$.
The discount factor $\gamma=0.1$, reinforcement interval $T=0.001s$.
The weight matrix for augmented states and control are  $Q_1=diag(10,10,50,0,0,0)$ and $R=I_3$, respectively. 
The baseline learning rate $\alpha=16.1$, parameters for variable gain gradient descent are $q_2=.1,k_2=.01$ and actuator saturation is $u_m=90$ deg.  
A dithering noise of the form, $n(t)=2e^{-0.009t}(\sin(11.9t)^2\cos(19.5t)+\sin(2.2t)^2\cos(5.8t)+\sin(1.2t)^2\cos(9.5t)+\sin(2.4t)^5)$ is added to maintain the persistent excitation (PE) condition as demonstrated in \cite{vamvoudakis2014online}. All the critic weights were initialized to $0$, i.e., $\hat{W}(0)=0$.
Desired set point for $\phi,\theta,\psi$ was set to $(-30,0,-10)$ degrees for first $7$ seconds, then $(30,0,10)$ degrees from 7-11 seconds and finally $(0,0,10)$ degrees.
Fig. \ref{fig:att} shows
that the proposed control scheme can accurately track the desired attitude. 
All the states of the UAV are bounded as can be inferred from Fig. (\ref{fig:states}). 
From Fig. \ref{fig:ctr}), the generated optimal control action is within $\pm90$ degrees. 
The weights of the critic NN converge close to their ideal values in finite amount of time as can be seen from Fig. \ref{fig:crit}.
Since the resultant cost is very small (close to $0$), as can be seen in Fig. \ref{fig:val}, the control policies corresponding to elevator, aileron and rudder (Fig. \ref{fig:ctr}) are approximately optimal.
Since the proposed update law is based on IRL framework, it does not require drift dynamics and hence a lot of aerodynamic stability and damping derivatives are not needed for control implementation.
The only aerodynamic coefficients that are required for this control strategy are the control derivatives that appear in control coupling dynamics 
(refer to page no. 61-62 of \cite{beard2012small}).
Further, critic-only approximation structure helps in reducing computational load as only one set of NN weights need to be updated at any given time instead of two NNs (actor-critic).
This is especially advantageous in applications, where lighter computational load is helpful in meeting stringent real-time requirements.

\section{Conclusion}\label{conclusion}
\vspace{-.1cm}
A novel parameter update law for single-approximator structure (critic-only neural network (NN)) in integral reinforcement learning framework to solve optimal trajectory tracking problem of partially-unknown continuous time nonlinear system with actuator constraints has been proposed in this paper. 
The presented update law has been shown to ensure the uniform ultimate boundedness (UUB) stability of the augmented system. The stabilizing term in the update law has helped in obviating the requirement of initial stabilizing controller. 
Moreover, the use of variable gain gradient descent in the presented update law could adjust the learning rate depending on the Hamilton-Jacobi-Bellman (HJB) error and instantaneous rate of variation of Lyapunov function along the system trajectories, which has ensured tighter residual set.
The presented algorithm has been successfully validated on full 6-DoF UAV model. 

\bibliographystyle{IEEEtran}
\bibliography{sample}             
                                                   







\end{document}